\documentclass[journal]{IEEEtran}

\ifCLASSINFOpdf
\else
\fi

\usepackage{color}
\usepackage{epsfig}
\usepackage{subfig}
\usepackage{mathrsfs}
\usepackage{booktabs}
\usepackage{graphicx}
\usepackage{epstopdf}
\usepackage{multirow}
\usepackage{amsmath,amssymb}
\usepackage{amsthm}
\usepackage{fancyhdr}
\usepackage{tabularx}
\usepackage{xfrac}
\usepackage{CJK}
\usepackage{mathptmx}
\usepackage{mathrsfs}
\usepackage{graphics}
\usepackage{times}

\newtheorem{prop}{Proposition}[section]

\newtheorem{remark}{Remark}[section]

\begin{document}

\title{Towards a Universal Approach for Identifying Cascading Failures of Power Grids}

\author{Chao Zhai, Gaoxi Xiao, Hehong Zhang and Tso-Chien Pan \thanks{Chao Zhai, Gaoxi Xiao, Hehong Zhang and Tso-Chien Pan are with Institute of Catastrophe Risk Management, Nanyang Technological University, 50 Nanyang Avenue, Singapore 639798. They are also with Future Resilient Systems, Singapore-ETH Centre, 1 Create Way, CREATE Tower, Singapore 138602. Chao Zhai, Gaoxi Xiao and Hehong Zhang are also with School of Electrical and Electronic Engineering, Nanyang Technological University, Singapore. Corresponding author: Gaoxi Xiao. Email: egxxiao@ntu.edu.sg}}


\maketitle

\begin{abstract}
Due to the evolving nature of power systems and the complicated coupling relationship of power devices, it has been a great challenge
to identify the contingencies that could trigger cascading blackouts of power systems. This paper aims to develop a universal approach for identifying the initial disruptive contingencies that can result in the worst-case cascading failures of power grids. The problem of contingency identification is formulated in a unified mathematical framework, and it can be solved by the Jacobian-Free Newton-Krylov (JFNK) method in order to circumvent the Jacobian matrix and relieve the computational burden. Finally, numerical simulations are carried out to validate the proposed identification approach on the IEEE $118$ Bus System.
\end{abstract}

\begin{IEEEkeywords}
Cascading failures, contingency identification, power grids, JFNK method
\end{IEEEkeywords}

\IEEEpeerreviewmaketitle

\section{Introduction}
\IEEEPARstart{T}{he} past decades have witnessed several large blackouts in the world such as India Blackout (2012), US-Canada Blackout (2003), Italy Blackout (2003) and Southern Brazil Blackout (1999) to name just a few, which have left millions of residents without power supply and caused huge financial losses \cite{mcl09}. In such catastrophe events, the initial contingencies ($e.g.$ extreme weather, terrorist attack and operator error) play a crucial role in triggering the cascading outage of power systems. It is reported that the mal-operation of a protection relay is the key ``trigger" of the final line outage sequence in most blackouts \cite{beck05}. For instance, conventional relays may lead to unselective tripping under high load conditions, which could initiate the chain reaction of branch outages under certain conditions (e.g., a wrong relay operation of Sammis-Star line in the 2003 US-Canada Blackout \cite{beck05}). The reliability and resilience of power grids are closely related to the proactive elimination of disruptive initial contingencies. Thus, it is vital to identify the initial contingency that causes the most severe blackouts and work out remedial schemes against cascading blackouts in advance.

In practice, electrical power devices such as FACTS devices, HVDC links and protective relays serve as the major protective barrier against cascading blackouts. To be specific, FACTS devices significantly contribute to the stability improvement of power systems, while HVDC links behave like a ``firewall" to prevent the propagation of cascading outages. Actually, the FACTS devices have been widely installed in power transmission networks to improve the capability of power transmission, controllability of power flow, damping of power oscillation and post-contingency stability. As a series FACTS device, the thyristor-controlled series capacitor (TCSC) allows fast and continuous adjustments of branch impedance in order to control the power flow and improve the transient stability \cite{jov05}. In addition, the HVDC links assist in preventing cascades propagation and restoring the power flow after faults. For example, Qu\'{e}bec power system in Canada survived the cascades in the 2003 US-Canada Blackout due to its DC interconnection to the US power systems \cite{beck05}. As the most common protection device, protective relays of power system react passively to the system oscillation and promptly remove the overloading elements without affecting the normal operation of the rest of the system. Meanwhile it allows for time delay of abnormal oscillations to neglect the trivial disturbances and avoid the overreaction to the transient state changes \cite{jia16}. It is necessary to take into account the protection mechanism of power devices for the practical cascading process.

So far, cascading failures of power systems have been investigated through two distinct routes. Specifically, some researchers propose rigorous mathematical formulations for exploring vulnerable elements or contingencies in power grids \cite{alm15,tae16}, while others focus on the accurate modeling of practical cascading failures \cite{jia16,yan15}. The identification approaches are developed to search for critical branches or initial malicious disturbances that can cause the large-scale disruptions \cite{chen05,dav11,don08,bie10,roc11,epp12,zhai17}. For instance, the methods are proposed to identify the collections of $n-k$ contingencies based on the event trees \cite{chen05}, line outage distribution factor \cite{dav11} and other optimization techniques \cite{don08,bie10,roc11}. Nevertheless, these optimization approaches are not efficient to identify the large collections of $n-k$ contingencies that result in cascading blackouts. To address this problem, a ``random chemistry" algorithm is designed with the relatively low computational complexity \cite{epp12}. In addition, an optimal control approach is adopted to identify the initial contingencies by treating these contingencies as the control inputs \cite{zhai17}. And this approach is able to determine the continuous changes of branch admittance other than direct branch outages as the initial contingencies. An adaptive algorithm based on multi-agent system is designed to prevent cascading failures of power grids subject to $N-1$ or $N-1-1$ contingencies without load shedding \cite{bab18}. It is suggested that the structural characteristic of communication networks and the interaction between power grids and communication networks are closely related to the ability of power grids to prevent the cascades \cite{cai16}.

This work aims at a rigorous mathematical formulation of identifying the worst-case cascading failures, which is expected to allow for the practical physical characteristics of power system cascades. Compared to existing work, the key contributions of this paper are highlighted as follows
\begin{enumerate}
  \item Propose a general mathematical formulation for identifying the various contingencies that can trigger the power system cascades and result in the severe disruptions.
  \item Develop an efficient numerical algorithm based on the Jacobian-Free Newton-Krylov method to search for the critical contingencies with the guaranteed performance in theory.
  \item Validate the proposed approach on a large-scale power grid and investigate the effect of time delay in protective relays on the cascading failures.
\end{enumerate}

The remainder of this paper is organized as follows. Section \ref{sec:prob} presents the cascades model and optimization formulation, followed by the numerical solver and theoretical analysis in Section \ref{sec:num}. Next, the identification approach is validated in Section \ref{sec:sim}. Finally, we conclude the paper and discuss future work in Section \ref{sec:con}.

\section{Problem Formulation}\label{sec:prob}

\subsection{Cascades Process}
This section aims to characterize the cascading failure of power grids subject to the initial contingency and system stresses. Figure \ref{cascade} presents the cascading process of power grids after the initial contingency or triggering event is added on the system.
Specifically, the initial contingency or triggering event may result in the independent branch outages, which could cause hidden failures in power grids or situational awareness errors of operators in the control center. As a result, the stresses in power networks are aggravated by hidden failures, mis-operations of operators or the independent branch outages directly. Such stresses may give rise to dependent branch outages further and thus force power systems to take protective actions (e.g., load shedding or generator tripping). Moreover, the branch outages change the configuration of power networks (i.e., network topology), which in return aggravates the stress of power networks. Essentially, the above positive feedback process contributes to accelerating the cascading failures and ends up with the power systems blackout.

\begin{figure}
\scalebox{0.04}[0.04]{\includegraphics{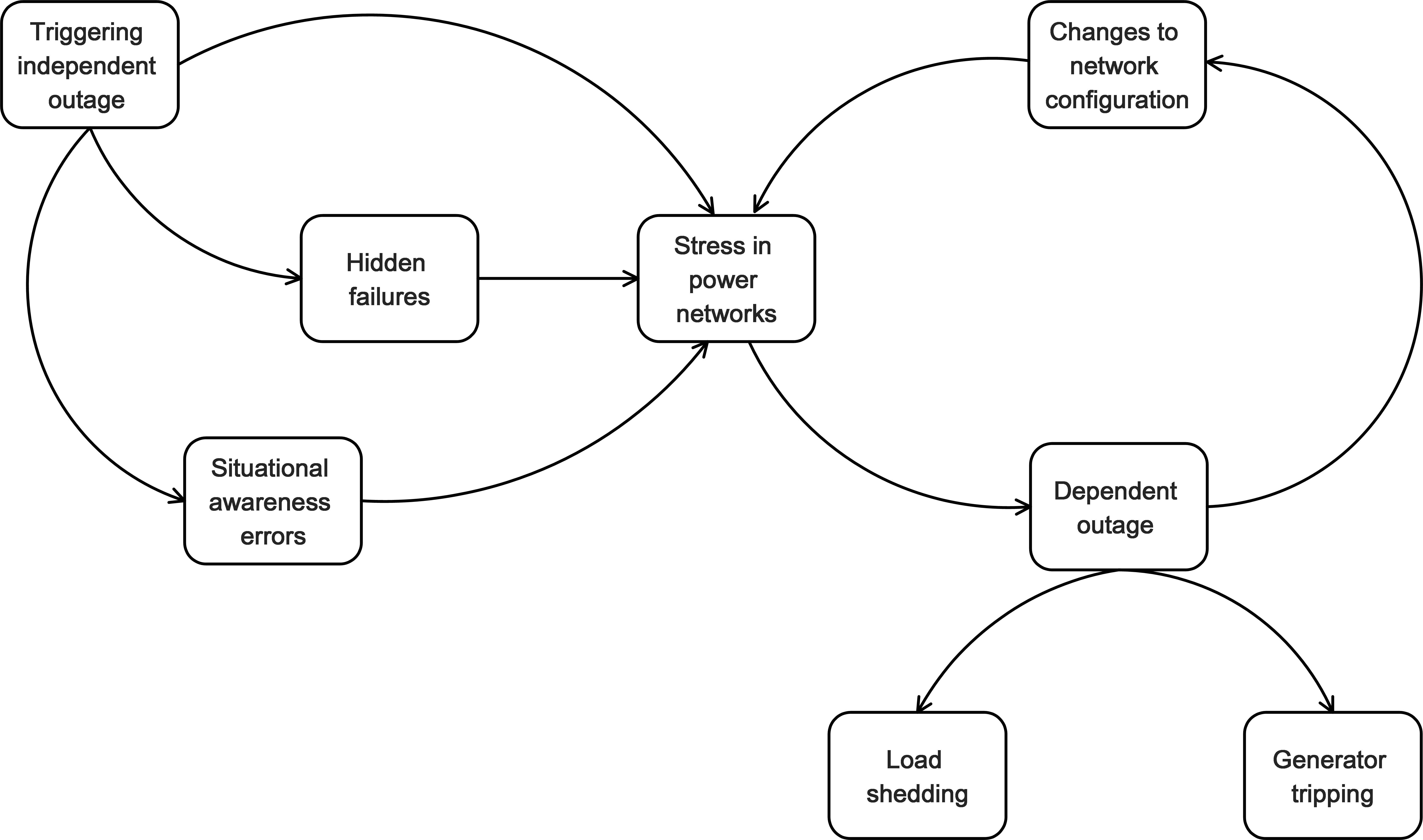}}\centering
\caption{\label{cascade} Cascading failure process of power grids \cite{hine16}.}
\end{figure}

\subsection{Mathematical model}
In terms of the practical cascades process, it is suggested that dynamic power networks can be modeled as a system of hybrid differential-algebraic equations \cite{jia16}.
\begin{equation}\label{hybrid}
    \left\{
      \begin{array}{ll}
        \dot{x}=f(t,x,y,\theta) & \hbox{} \\
        0=g(t,x,y,\theta) & \hbox{} \\
        0>h(t,x,y,\theta) & \hbox{}
      \end{array}
    \right.
\end{equation}
where $x$ denotes a vector of continuous state variables subject to differential relationships, and $y$ represents a vector of continuous state variables under the constraints of algebraic equations. In addition, $\theta$ refers to a vector of discrete binary state variables (i.e., $\theta_i\in\{0,1\}$). The set of differential equations in the system (\ref{hybrid}) characterize the dynamic response of machines, governors, exciters and loads in power grids. The algebraic components mainly describe the AC power flow equations, and the inequality terms reflect the discrete events (e.g., the automatic line tripping by protective relays, manual operations, lightning, etc) during cascading failures. In practice, the structure of power grids (e.g., network topology, component parameters) is affected once a discrete event occurs. Thus, the discrete events directly influence the dynamic response of relevant devices and the distribution of power flow in power grids. To incorporate the effect of discrete events at time instants $t_k$, the time axis is divided into a series of time intervals $[t_{k-1},t_k)$, $k\in I_m=\{1,2,...,m\}$. At each time interval, the set of differential-algebraic equations is solved using the updated parameters and initial conditions of power system model due to discrete events. By solving the system (\ref{hybrid}) in each time interval, the vectors of state variables $x$ and $y$ at the terminal of each time interval can be obtained as follows
\begin{equation}\label{evolu}
    \left\{
      \begin{array}{ll}
        x_k=F(t_k,x_{k-1},y_{k-1},\theta_{k-1}) & \hbox{} \\
        y_k=G(t_k,x_{k-1},y_{k-1},\theta_{k-1}) & \hbox{}
      \end{array}
    \right.
\end{equation}
where $x_k=x(t_k)$, $y_k=y(t_k)$ and $\theta_k=\theta(t_k)$, $k\in I_m$. And the iterated functions $F$ and $G$ describe the evolution process of state variables $x_k$ and $y_k$, respectively.

\subsection{Optimization formulation}
The cascading blackouts result in the severe damage of power networks and paralyze the service of power supply. Our goal is to search for the initial contingencies that cause the worst disruptions of power grids at the end of cascading blackouts. Therefore, the problem of identifying initial contingencies in power grids is formulated as
\begin{equation}\label{formulation}
\begin{split}
&~~~~~~~\min_{\delta\in\Omega}~ J(\delta, x_m, y_m) \\
&s.~t.~x_k=F(t_k,x_{k-1},y_{k-1},\theta_{k-1}) \\
&~~~~~~~y_k=G(t_k,x_{k-1},y_{k-1},\theta_{k-1}),~k \in I_m \\
\end{split}
\end{equation}
where $\delta$ denotes the initial contingencies in power grids that change the state variables $x$, $y$ or $z$ in the initial time interval $[t_0,t_1)$. And $\Omega$ represents the set of physical constraints on the initial contingencies, and it can be described as $\bigcap_{i=1}^{n}\{\delta~|~v_i(\delta)\leq0\}$ with the inequality constraints $v_i(\delta)\leq0$.  For simplicity, it is assumed that the triggering event or initial contingency occurs at time $\tau\in[t_0,t_1)$. Then we have $\left(x(\tau_{+}),y(\tau_{+}),\theta(\tau_{+})\right)=\Gamma(x(\tau),y(\tau),\theta(\tau),\delta)$, and the function $\Gamma$ characterizes the effect of the contingency $\delta$ on the state variables at time $\tau$.

The objective function $J(\delta, x_m, y_m)$ quantifies the disruptive level of power grids at the end of cascading failures. A smaller value of $J(\delta, x_m, y_m)$ indicates a worse disruption of power grids due to cascading blackouts. Then it follows from the Karush-Kuhn-Tucker (KKT) conditions that the necessary conditions for optimal solutions to Optimization Problem (\ref{formulation}) is presented as follows \cite{man94}.
\begin{prop}
The optimal solution $\delta^{*}$ to Optimization Problem (\ref{formulation}) with the multipliers $\mu_i$, $i\in I_n$ satisfies the KKT conditions
\begin{equation}\label{kkt}
\begin{split}
&\nabla J(\delta^*,x_m,y_m)+\sum_{i=1}^{n}\mu_i\nabla v_i(\delta^*)=\mathbf{0} \\
&v_i(\delta^*)+\omega_i^2=0 \\
&\mu_i\cdot v_i(\delta^*)=0 \\
&\mu_i-\sigma_i^2=0, \quad i\in I_n \\
\end{split}
\end{equation}
where $\omega_i$ and $\sigma_i$, $i\in I_n$ are the unknown variables.
\end{prop}

\begin{proof}
The KKT conditions for Optimization Problem (\ref{formulation}) are composed of four components: stationary, primal feasibility, dual feasibility and complementary slackness. Specifically,
stationary condition allows us to obtain
$$
\nabla J(\delta^*,x_m,y_m)+\sum_{i=1}^{n}\mu_i\nabla v_i(\delta^*)=0,
$$
where
$$
\Omega=\bigcap_{i=1}^{n}\{\delta~|~v_i(\delta)\leq 0\}.
$$
Moreover, the primal feasibility leads to $g_i(\delta)\leq0$, $i\in I_n$, which can be converted into equality constraints
$$
v_i(\delta^*)+\omega_i^2=0, \quad i\in I_n
$$
with the unknown variables $\omega_i\in R$. Further, the dual feasibility corresponds to $\mu_i\geq0$, which can be replaced by
$$
\mu_i-\sigma_i^2=0, \quad i\in I_n
$$
with the unknown variables $\sigma_i\in R$. Finally, the complementary slackness gives
$$
\mu_i\cdot v_i(\delta^*)=0, \quad i\in I_n
$$
This completes the proof.
\end{proof}

\begin{remark}
To reduce the computation burden, the gradient $\nabla J(\delta^*,x_m,y_m)$ can be approximated by
\begin{equation}\label{app}
\begin{split}
\nabla J(\delta,x_m,y_m)|_{\delta=\delta^*}&=\left(\frac{\partial J(\delta,x_m,y_m)}{\partial\delta_i}|_{\delta=\delta^*}\right) \in R^{dim(\delta)} \\
&\approx \left(\frac{J(\delta^*+\epsilon\cdot e_i,x_m,y_m)-J(\delta^*,x_m,y_m)}{\epsilon}\right)
\end{split}
\end{equation}
with the sufficiently small $\epsilon$ and the unit vector $e_i$ with $1$ in the $i$-th position and $0$ elsewhere. And the symbol $dim(\delta)$ denotes the dimension of the variable $\delta$.
\end{remark}

\section{Numeric Solver} \label{sec:num}
To avoid the computation of partial derivatives in (\ref{app}), the Jacobian Free Newton Krylov (JFNK) method is employed in this section to solve the system of nonlinear algebraic equations without computing the Jacobian matrix. Essentially, the JFNK methods are synergistic combinations of Newton methods for solving nonlinear equations and Krylov subspace methods for solving linear equations \cite{kno04}.
To facilitate the analysis, the system (\ref{kkt}) is rewritten in matrix form
\begin{equation}\label{mequation}
    S(\mathbf{z})=\mathbf{0},
\end{equation}
where the unknown vector $\mathbf{z}$ is composed of $\delta^*$, $\mu_i$, $\omega_i$, $\sigma_i$, $i\in I_n$. And $\mathbf{0}$ refers to a zero vector with the proper dimension. To obtain the iterative formula for solving (\ref{mequation}), the Taylor series of $S(\mathbf{z})$ at $\mathbf{z}^{s+1}$ is computed as follows
\begin{equation}\label{taylor}
    S(\mathbf{z}^{s+1})=S(\mathbf{z}^{s})+\mathfrak{J}(\mathbf{z}^s)(\mathbf{z}^{s+1}-\mathbf{z}^{s})+O(\Delta\mathbf{z}^{s})
\end{equation}
with $\Delta\mathbf{z}^{s}=\mathbf{z}^{s+1}-\mathbf{z}^{s}$. By neglecting the high-order term $O(\Delta\mathbf{z}^{s})$ and setting $S(\mathbf{z}^{s+1})=\mathbf{0}$, we obtain
\begin{equation}\label{lequ}
    \mathfrak{J}(\mathbf{z}^s)\cdot\Delta\mathbf{z}^{s}=-S(\mathbf{z}^{s}),\quad  \quad s\in Z^{+}
\end{equation}
where $\mathfrak{J}(\mathbf{z}^s)$ represents the Jacobian matrix and $s$ denotes the iteration index. Thus, solutions to Equation (\ref{mequation}) can be approximated by implementing Newton iterations $\mathbf{z}^{s+1}=\mathbf{z}^{s}+\Delta\mathbf{z}^{s}$,
where $\Delta\mathbf{z}^{s}$ is obtained by Krylov methods. First of all, the Krylov subspace is constructed as follows
\begin{equation}\label{kspace}
    K_i=\mathrm{span}\left(\mathbf{r}^s,~\mathfrak{J}(\mathbf{z}^s)\mathbf{r}^s,~\mathfrak{J}(\mathbf{z}^s)^2\mathbf{r}^s,...,~\mathfrak{J}(\mathbf{z}^s)^{i-1}\mathbf{r}^s\right)
\end{equation}
with $\mathbf{r}^s=-S(\mathbf{z}^{s})-\mathfrak{J}(\mathbf{z}^s)\cdot\Delta\mathbf{z}_0^{s}$, where $\Delta\mathbf{z}_0^{s}$ is the initial guess for the Newton correction and is typically zero \cite{kno04}. Actually, the optimal solution to $\Delta\mathbf{z}^{s}$ is the linear combination of elements in Krylov subspace $K_i$.
\begin{equation}\label{gmres}
    \Delta \mathbf{z}^s=\Delta \mathbf{z}_0^s+\sum_{j=1}^{i-1}\lambda_j \cdot \mathfrak{J}(\mathbf{z}^s)^j\mathbf{r}^s,
\end{equation}
where $\lambda_j$, $j\in\{1,2,...,i-1\}$ is obtained by minimizing the residual $\mathbf{r}^s$ with the Generalized Minimal RESidual (GMRES) method with the constraint of step size $\|\Delta \mathbf{z}^s\|\leq c$ \cite{saad86}. In particular, matrix-vector products in (\ref{gmres}) can be approximated by
\begin{equation}\label{jv}
\mathfrak{J}(\mathbf{z}^s)\mathbf{r}^s\approx\frac{S(\mathbf{z}^s+\epsilon\cdot\mathbf{r}^s)-S(\mathbf{z}^s)}{\epsilon},
\end{equation}
where $\epsilon$ is a sufficiently small value \cite{saad90}. In this way, the computation of Jacobian matrix is avoided via matrix-vector products in (\ref{jv}) while solving Equation (\ref{mequation}). Actually, the accuracy of the forward difference scheme (\ref{jv}) can be estimated as follows.
\begin{prop}
$$
\left\|\frac{S(\mathbf{z}^s+\epsilon\cdot\mathbf{r}^s)-S(\mathbf{z}^s)}{\epsilon}-\mathfrak{J}(\mathbf{z}^s)\mathbf{r}^s\right\|\leq\frac{\epsilon\|\mathbf{r}^s\|^2}{2}\sup_{t\in[0,1]} \|S^{(2)}(\mathbf{z}^s+t\epsilon\cdot\mathbf{r}^s)\|
$$
where $S^{(2)}(z)$ denotes the second order derivative of $S(z)$ with respect to the variable $z$.
\end{prop}

\begin{proof}
It follows from NR $3.3$-$3$ in \cite{ort00} that
$$
\frac{S(\mathbf{z}^s+\epsilon\cdot\mathbf{r}^s)-S(\mathbf{z}^s)}{\epsilon}-\mathfrak{J}(\mathbf{z}^s)\mathbf{r}^s=\int_{0}^{1}\epsilon (1-t)S^{(2)}(\mathbf{z}^s+t\epsilon\cdot\mathbf{r}^s)\mathbf{r}^s\mathbf{r}^sdt
$$
which implies
\begin{equation*}
    \begin{split}
         &~~~\left\|\frac{S(\mathbf{z}^s+\epsilon\cdot\mathbf{r}^s)-S(\mathbf{z}^s)}{\epsilon}-\mathfrak{J}(\mathbf{z}^s)\mathbf{r}^s\right\|\\
         &=\left\|\int_{0}^{1}\epsilon (1-t)S^{(2)}(\mathbf{z}^s+t\epsilon\cdot\mathbf{r}^s)\mathbf{r}^s\mathbf{r}^sdt\right\| \\
         &\leq \epsilon \int_{0}^{1}(1-t)\left\|S^{(2)}(\mathbf{z}^s+t\epsilon\cdot\mathbf{r}^s)\mathbf{r}^s\mathbf{r}^s\right\|dt\\
         &\leq \epsilon \int_{0}^{1} (1-t)\|S^{(2)}(\mathbf{z}^s+t\epsilon\cdot\mathbf{r}^s)\|\cdot\|\mathbf{r}^s\|^2dt \\
         & \leq \epsilon \sup_{t\in[0,1]}\|S^{(2)}(\mathbf{z}^s+t\epsilon\cdot\mathbf{r}^s)\|\cdot\|\mathbf{r}^s\|^2\int_{0}^{1}(1-t)dt \\
         & = \frac{\epsilon}{2}\|\mathbf{r}^s\|^2\cdot\sup_{t\in[0,1]} \|S^{(2)}(\mathbf{z}^s+t\epsilon\cdot\mathbf{r}^s)\|
     \end{split}
\end{equation*}
The proof is thus completed.
\end{proof}

\begin{remark}
The choice of $\epsilon$ greatly affects the accuracy and robustness of the JFNK method. For the forward difference scheme (\ref{jv}), $\epsilon$ can be set equal to a value larger than the square root of machine epsilon to minimize the approximation error \cite{chan84}.
\end{remark}

\begin{table}
 \caption{\label{cia} Contingency Identification Algorithm.}
 \begin{center}
 \begin{tabular}{lcl} \hline
 \textbf{Initialize:} $l_{\max}$, $\epsilon_{\min}$, $\epsilon_{0}$, and $\delta=\bf{0}$ \\
 \textbf{Goal:} $\delta^*$ and $J(\delta^*,x_m,y_m)$ \\ \hline
  1: \textbf{for}~$l=0$ \textbf{to} $l_{\max}$ \\
  2: ~~~~~~~~$s=0$ \\
  3: ~~~~~~~\textbf{while}~($\epsilon_{s}>\epsilon_{\min}$) \\
  4: ~~~~~~~~~~~~~~Calculate the residual $\mathbf{r}^s=-S(\mathbf{z}^{s})-\mathfrak{J}(\mathbf{z}^s)\cdot\Delta\mathbf{z}_0^{s}$ \\
  5: ~~~~~~~~~~~~~~Construct the Krylov subspace $K_i$ in (\ref{kspace}) \\
  6: ~~~~~~~~~~~~~~Approximate $\mathfrak{J}(\mathbf{z}^s)^j\mathbf{r}^s$ in (\ref{gmres}) using (\ref{jv})   \\
  7: ~~~~~~~~~~~~~~Compute $\lambda_j$ in (\ref{gmres}) with the GMRES method \\
  8: ~~~~~~~~~~~~~~Compute $\Delta\mathbf{z}^{s}$  with (\ref{gmres}) \\
  9: ~~~~~~~~~~~~$\mathbf{z}^{s+1}=\mathbf{z}^{s}+\Delta\mathbf{z}^{s}$ \\
  10: ~~~~~~~~~~~~$\epsilon_{s+1}=\|\Delta\mathbf{z}^{s}\|/\|\mathbf{z}^s\|$ \\
  11: ~~~~~~~~~~~~$s=s+1$ \\
  12: ~~~~~~~\textbf{end while} \\
  13: ~~~~~~~~Update $\delta^*$  and $J(\delta^*,x_m,y_m)$ \\
  14: ~~~~~~~~\textbf{if} $\left(J(\delta^*,x_m,y_m)<J(\delta,x_m,y_m)\right)$ \\
  15: ~~~~~~~~~~~$\delta=\delta^*$ \\
  16: ~~~~~~~~\textbf{end if} \\
  17: ~~~~~~~ $l=l+1$ \\
  18: \textbf{end for} \\ \hline
 \end{tabular}
 \end{center}
\end{table}

Table \ref{cia} presents the implementation process of Contingency Identification Algorithm (CIA) with the aid of the JFNK method.
First of all, the initial values for some variables are specified as follow: $\delta=0$ and $l=0$, $\epsilon_{min}$, $\epsilon_0$ with
the condition $\epsilon_{\min}<\epsilon_{0}$, and the maximum iterative step $l_{\max}$. Then the JFNK method is employed to obtain the optimal disturbance $\delta^*$ and the cost $J(\delta^*,x_m,y_m)$ from Step $4$ to Step $13$. Specifically, the residual $\mathbf{r}^s$ is calculated in each iteration in order to construct the Krylov subspace $K_i$. For elements in $K_i$, the matrix-vector products are approximated by Equation (\ref{jv}) without forming the Jacobian. Next, the term $\Delta\mathbf{z}^{s}$ for Newton iterations is obtained via the GMRES method. The tolerance $\epsilon_s$ and step number $s$ are updated after implementing the Newton iteration for $\mathbf{z}^s$. Afterwards, a new iteration loop is launched if the termination condition $\epsilon_{s}\leq\epsilon_{\min}$ fails. After adopting the JFNK method, a disturbance value $\delta^*$ in (\ref{kkt}) is saved if it results in a worse cascading failure ($i.e.$, $J(\delta^*,x_m,y_m)<J(\delta,x_m,y_m)$). The above algorithm does not terminate until the maximum iterative step $l_{\max}$ is reached.

The following theoretical results allow us to roughly estimate the convergence accuracy of initial disturbances before implementing the CIA.
\begin{prop}
With the Contingency Identification Algorithm in Table \ref{cia}, the increment $\Delta\delta$ is upper bounded by
$$
\|\Delta\delta\| \leq \epsilon_{\min} \cdot \left(\|z^0\|+c \cdot s_{\max}\right)
$$
where $z^0$ denotes the initial value for the unknown vector $z$ in the numerical algorithm, and $s_{\max}$ refers to the maximum iteration steps.
\end{prop}

\begin{proof}
According to the Contingency Identification Algorithm, we have the following inequality
$$
\frac{\|\Delta\mathbf{z}^{s}\|}{\|\mathbf{z}^s\|}\leq\epsilon_{\min}
$$
after adopting the JFNK method. In addition, it follows from the updating law $\mathbf{z}^{s+1}=\mathbf{z}^{s}+\Delta\mathbf{z}^{s}$ that
$\mathbf{z}^{s}=\mathbf{z}^{0}+\sum_{i=0}^{s-1}\Delta\mathbf{z}^{i}$, which allows us to obtain
\begin{equation*}
    \begin{split}
    \|\Delta\mathbf{z}^{s}\|&\leq\epsilon_{\min} \cdot \|\mathbf{z}^s\|\\
                            &=\epsilon_{\min} \cdot \left\|\mathbf{z}^{0}+\sum_{i=0}^{s-1}\Delta\mathbf{z}^{i}\right\| \\
                            &\leq \epsilon_{\min}\cdot \left(\|\mathbf{z}^{0}\|+\sum_{i=0}^{s-1}\|\Delta\mathbf{z}^{i}\|\right) \\
                            & \leq \epsilon_{\min} \cdot \left(\|z^0\|+c \cdot s_{\max}\right),
    \end{split}
\end{equation*}
due to $\|\Delta \mathbf{z}^s\|\leq c$ and $s\leq s_{\max}$. Moreover, it follows from $\|\Delta\delta\|\leq\|\Delta\mathbf{z}^{s}\|$ that we have
$$
\|\Delta\delta\|\leq\epsilon_{\min} \cdot \left(\|z^0\|+c \cdot s_{\max}\right),
$$
which completes the proof.
\end{proof}

\begin{remark}
According to the CIA in Table \ref{cia}, the value of cost function $J(\delta^*,x_m,y_m)$ decreases monotonically as the iteration step $l_{\max}$ increases. Considering that $J(\delta^*,x_m,y_m)$ is normally designed to have a lower bound (i.e., $J(\delta^*,x_m,y_m)\geq0$), $J(\delta^*,x_m,y_m)$ converges to a local minimum. This enables us to identify the corresponding initial disturbances $\delta^*$.
\end{remark}

\section{Case Study}\label{sec:sim}
In this section, the proposed CIA in Table \ref{cia} is implemented to search for the disruptive disturbances on selected branches of IEEE 118 Bus System \cite{zim11}. Numerical results on disruptive disturbances are validated by disturbing the selected branch with the magnitude of disturbance identified by the CIA.

\subsection{Cascades model}
In the simulations, a simple cascades model is taken into account and it includes FACTS devices, HVDC links and protective relays. The mathematical descriptions of these components are presented in the Appendix. In addition, the DC power flow equation is employed to ensure the computational efficiency and avoid the numerical non-convergence \cite{jia16}. When power grids are subject to the malicious disturbances, the FACTS devices take effect to adjust the branch admittance and balance the power flow for relieving the stress of power networks. If the stress is not eliminated, protective relays will be activated to serve the overloading branches on the condition that the timer of circuit breakers runs out of the preset time. The outage of overloading branches may result in the severer stress of power transmission networks and end up with the cascading blackout.
The evolution time of cascading failure is introduced to allow for the time factor of cascading blackouts. Essentially, the time interval between two consecutive cascading steps basically depends on the preset time of the timer in protective relays \cite{jia16}. Thus, the evolution time of cascading failure is roughly estimated by $t=kT$ at the $k$-th cascading step.

\subsection{Parameter setting}
Per-unit system is adopted with the base value of $100$ MVA in numerical simulations, and the power flow threshold for each branch is $5\%$ larger than the normal power floe on each branch without any disturbances. The power flow on each branch is close to the saturation, although it does not exceed their respective thresholds. In this way, the power system is vulnerable to initial contingencies, and thus is likely to suffer from cascading blackouts. The cost function in (\ref{formulation}) is designed as $\|P_e(\delta,P^m, Y^m_p)\|^2$ to minimize the total power flow on branches by identifying the initial disturbance $\delta$. Here $P_e$ represents the vector of power flow on branches. $P^m$ and $Y^m_p$ denote the vector of injected power on buses and that of branch admittance at the end of cascading failure, respectively. The maximum iterative step $l_{\max}$ is equal to $10$ in the CIA. Other parameters are given as follows: $\epsilon=10^{-2}$ in Equation (\ref{app}), $\epsilon_{\min}=10^{-8}$ in the JFNK method. Branch $8$ (i.e., the red link connecting Bus $5$ to Bus $8$ in Fig.~\ref{Step1}) is randomly selected as the disturbed element of power transmission networks. The lower and upper bounds of initial disturbances on Branch $8$ are given by $\underline{\delta}=0$ and $\bar{\delta}=37.45$, respectively. Actually, the upper bound of initial disturbances directly leads to the branch outage. And the total number of cascading steps is $m=12$.
For simplicity, we specify the same values for the parameters of three HVDC links as follows: $R_{ci}=R_{cr}=R_L=0.1$, $\alpha=\pi/15$ and $\gamma=\pi/4$. Regarding the FACTS devices, we set $X_{min,i}=0$, $X_{max,i}=10$ and $X^{*}_i=0$ for the TCSC, and $K_P=4$, $K_I=3$ and $K_D=2$ for the PID controller. In addition, the reference power flow $P^{*}_{e,i}$ is equal to the threshold of power flow on the relevant branch.

\begin{figure}\centering
 {\includegraphics[width=0.5\textwidth]{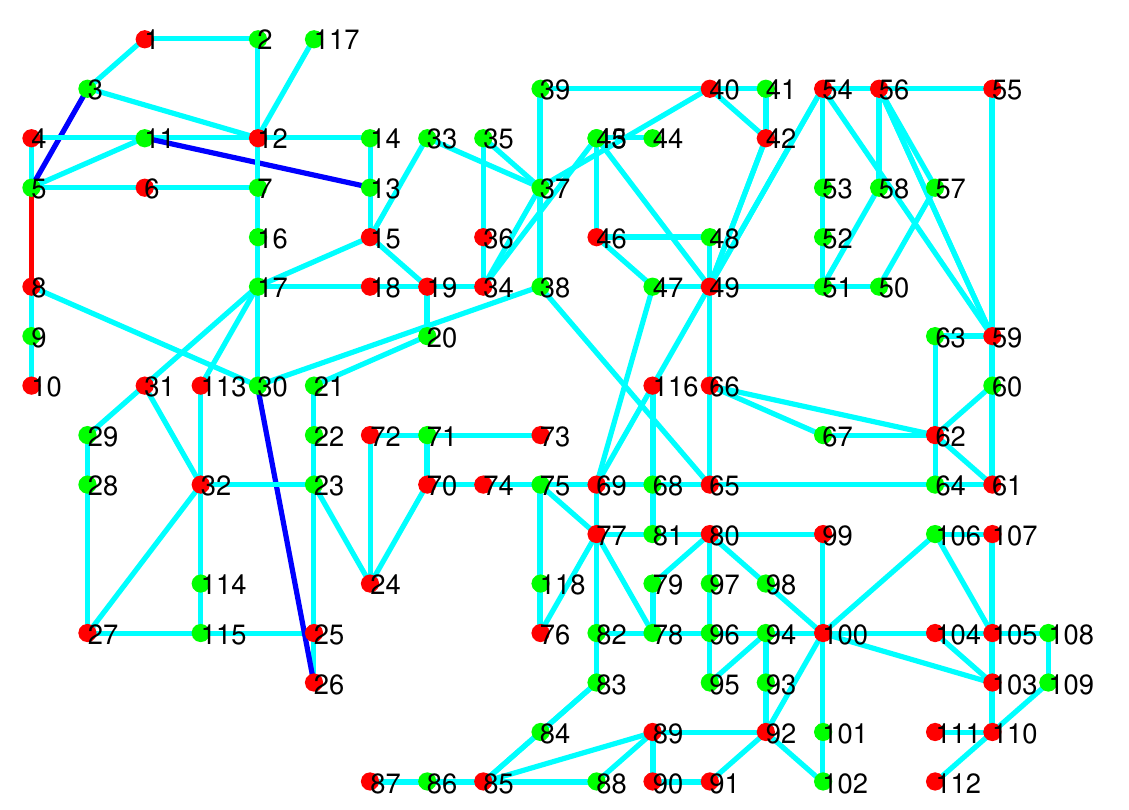}}
 \caption{\label{Step1} Initial state of IEEE 118 Bus System. Red balls denote the generator buses, while blue ones stand for the load buses. Cyan lines represent the branches of power systems. In addition, the red line is selected as the disturbed branch, and three blue lines are the HVDC links, including Branch 4, Branch 16 and Branch 38.}
\end{figure}

\subsection{Simulation and validation}
Figure \ref{Step1} shows the initial state of IEEE 118 Bus System in the normal condition, and this power system includes 53 generator buses, 64 load buses, 1 reference bus (i.e., Bus $69$) and 186 branches. And the HVDC links are denoted by blue lines, which include Branch 4 connecting Bus $3$ to Bus $5$, Branch $16$ connecting Bus $11$ to Bus $13$ and Branch $38$ connecting Bus $26$ to Bus $30$. Two preset values of the timer are taken into consideration in protective relays, $i.e.$, $T=0.5$s and $T=1$s. Contingency Identification Algorithm is carried out to search for the disturbance that results in the worst-case cascading failures of power systems (i.e., the minimum value of cost function $\|P_e(\delta,P^m, Y^m_p)\|^2$). For the IEEE 118 Bus System without the FACTS devices, the computed magnitude of disturbance on Branch $8$ is $37.45$, which exactly leads to the outage of Branch 8. For the power system with the FACTS devices and the preset time of the timer $T=0.5$s, the disturbance magnitude identified by the CIA is $36.77$, while it is $35.98$ for $T=1$s.

\begin{figure}\centering
 {\includegraphics[width=0.5\textwidth]{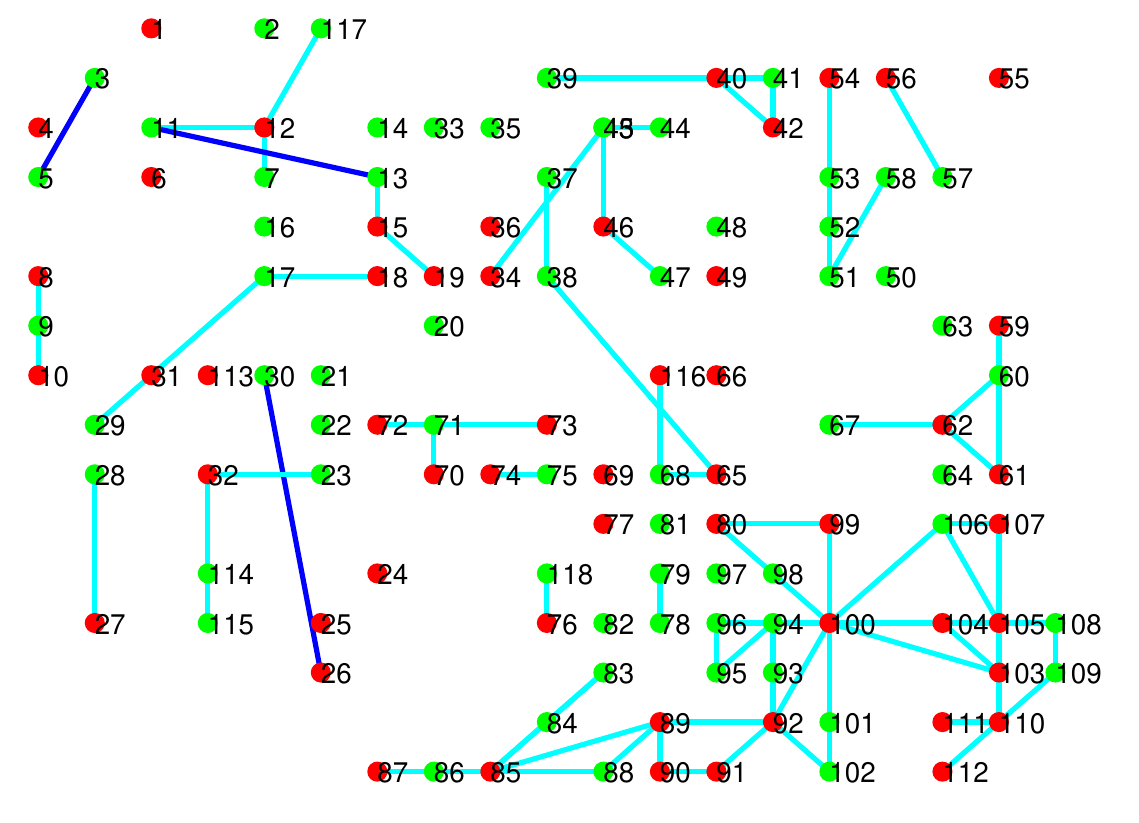}}
 \caption{\label{nofacts} Final configuration of IEEE 118 Bus System without FACTS devices.}
\end{figure}

Next, we validate the proposed identification approach by adding the computed disturbances on Branch 8 of IEEE $118$ Bus Systems. Specifically, Figure~\ref{nofacts} demonstrates the final state of IEEE 118 Bus System with no FACTS devices and with the preset time of circuit breaker $T=1$s. The cascading process terminates with 95 outage branches and the value of cost function is $53.28$ after 16 seconds, and the system collapses with 42 islands in the end. These 42 islands include $24$ isolated buses and $18$ subnetworks encircled by the dashed lines. In contrast, Figure~\ref{facts05} presents the final configuration of IEEE 118 Bus Systems with the protection of the FACTS devices and with the preset time $T=0.5$s. The cascading process ends up with 40 outage branches and the value of cost function
is $102.56$ after 10 seconds, and the power system is separated into $17$ islands, which include $6$ subnetworks and $11$ isolated buses. Figure \ref{facts1} gives the final state of power systems with FACTS devices and $T=1$s. It is observed that the power network is eventually split into $3$ islands (Bus 14, Bus 16 and a subnetwork composed of all other buses) with only $6$ outage branches and the cost function of $153.69$. Note that the initial disturbances identified by the CIA fail to cause the outage of Branch $8$ in the end for both $T=0.5$s and $T=1$s. The above simulation results demonstrate the advantage of the FACTS devices in preventing the propagation of cascading outages. A larger preset time of timer enables the FACTS devices to sufficiently adjust the branch admittance in response to the overload stress. As a result, the less severe damages are caused by the contingency for the larger preset time of timer.

\begin{figure}\centering
{\includegraphics[width=0.5\textwidth]{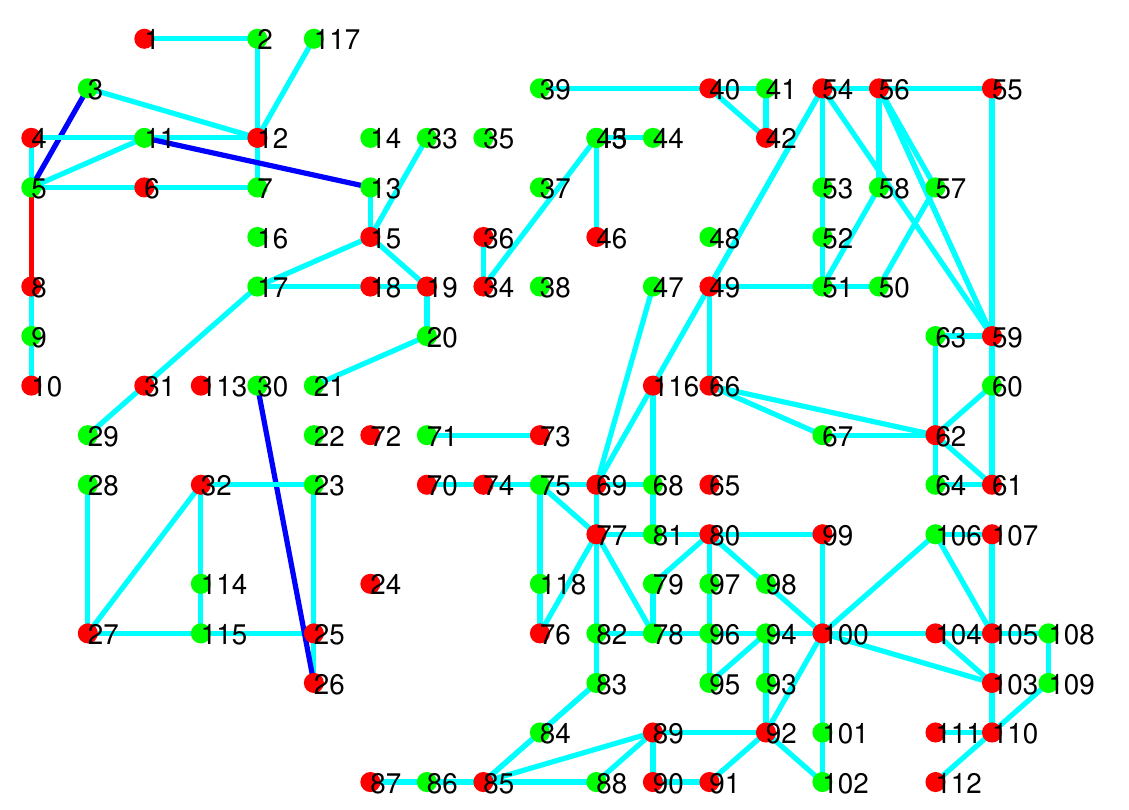}}
\caption{\label{facts05} Final configuration of IEEE 118 Bus System with FACTS devices and $T=0.5$s.}
\end{figure}

\begin{figure}\centering
{\includegraphics[width=0.5\textwidth]{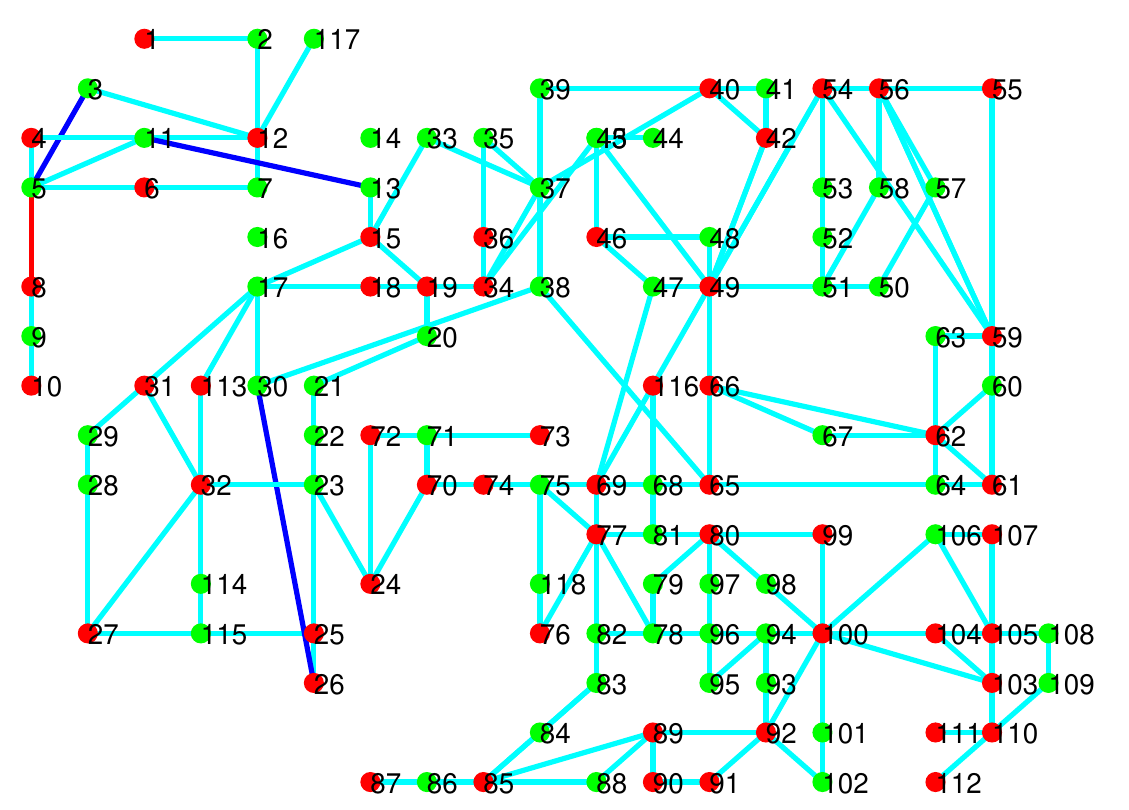}}
\caption{\label{facts1} Final configuration of IEEE 118 Bus System with FACTS devices and $T=1$s.}
\end{figure}

Figure \ref{comp} presents the time evolution of branch outages in the IEEE 118 Bus System as a result of disturbing Branch $8$ in three different scenarios. The cyan squares denote the number of outage branches with no FACTS devices and $T=1$s, while the green and blue ones refer to the numbers of outage branches with the FACTS devices and with $T=0.5$s and $T=1$s, respectively. The disturbances identified by the CIA are added to change the admittance of Branch $8$ at $t=0$s. With no FACTS devices, the cascading outage of branches propagates quickly from $t=2$s to $t=10$s and terminates at $t=16$s. When the FACTS devices are adopted and the preset time of timer is $T=0.5$s, the cascading failure starts at $t=2$s and speeds up till $t=8$s and stops at $t=10$s. For $T=1$s, the cascading outage propagates slowly due to the larger preset time of timer and comes to an end with only $6$ outage branches at $t=8$s. Together with protective relays and HVDC links, the FACTS devices succeed in protecting power systems against blackouts by adjusting the branch impedance in real time. More precisely, the number of outage branches decreases by $57.9\%$ with FACTS devices and $T=0.5$s and decreases by $93.7\%$ with FACTS devices and $T=1$s.

\begin{figure}\centering
{\includegraphics[width=0.5\textwidth]{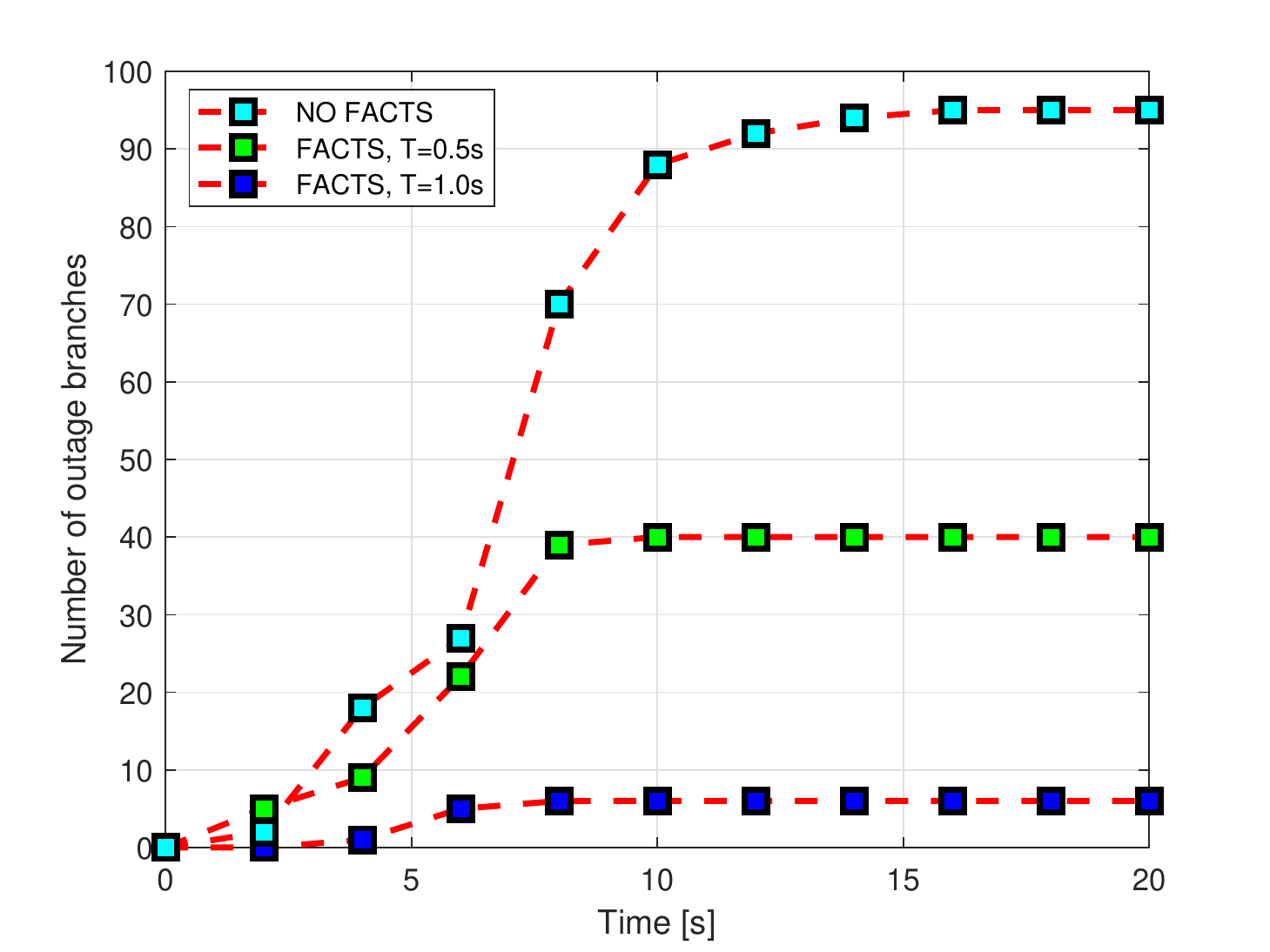}}
\caption{\label{comp} Time evolution of outage branches during cascading blackouts.}
\end{figure}

\section{Conclusions}\label{sec:con}

In this paper, we investigated the problem of identifying the initial contingencies that lead to cascading blackout of power transmission networks equipped with FACTS devices, HVDC links and protective relays. A universal optimization formulation was proposed to identify the contingencies, and an efficient numerical method was presented to solve the optimization problem. Numerical simulations were carried out on the IEEE 118 Bus Systems to validate the proposed identification approach. Significantly, the proposed contingency identification algorithm allows us to detect some nontrivial disturbances that lead to the severe cascading failure of power transmission networks, other than directly severing the branch. It is demonstrated that the coordination of FACTS devices and protective relays greatly enhances the capability of power grids against blackouts. Future work may include the comparison of different cascades models and the validation of the identified disturbances with real data in power system blackouts.

\section*{Acknowledgment}
This work is partially supported by the Future Resilience System Project at the Singapore-ETH Centre (SEC), which is funded by the National Research Foundation of Singapore (NRF) under its Campus for Research Excellence and Technological Enterprise (CREATE) program. It is also supported by Ministry of Education of Singapore under Contract MOE2016-T2-1-119.

\section*{Appendix: Component Models}


\subsection{FACTS devices}
FACTS devices can greatly enhance the stability and transmission capability of power systems. As an effective FACTS device, TCSC has been widely installed to control the branch impedance and relieve system stresses. The dynamics of TCSC is described by a first order dynamical model \cite{pas95}
\begin{equation}\label{tcsc}
    T_{C,i}\frac{d{X}_{C,i}}{dt}=-X_{C,i}+X^{*}_{i}+u_i, \quad X_{\min,i}\leq X_{C,i}\leq X_{\max,i}
\end{equation}
where $X^{*}_{i}$ refers to its reference reactance of Branch $i$ for the steady power flow. $X_{\min,i}$
and $X_{\max,i}$ are the lower and upper bounds of the branch reactance $X_{C,i}$ respectively and $u_i$ represents the supplementary control input, which is designed to stabilize the disturbed power system \cite{son00}. For simplicity, PID controller is adopted to regulate the power flow on each branch
\begin{equation}\label{pid}
    u_i(t)=K_P\cdot e_i(t)+K_I\cdot\int_{0}^{t}e_i(\tau)d\tau+K_D\cdot\frac{de_i(t)}{dt}
\end{equation}
where $K_P$, $K_I$ and $K_D$ are tunable coefficients, and the error $e_i(t)$ is given by
$$
e_i(t)=\left\{
         \begin{array}{ll}
           P^{*}_{e,i}-|P_{e,i}(t)|, & \hbox{$|P_{e,i}(t)|\geq P^{*}_{e,i}$;} \\
           0, & \hbox{otherwise.}
         \end{array}
       \right.
$$
Here, $P^{*}_{e,i}$ and $P_{e,i}(t)$ denote the reference power flow and the actual power flow on Branch $i$, respectively. Note that TCSC fails to function when the transmission line is severed.

\subsection{HVDC links}

HVDC links work as a protective barrier to prevent the propagation of cascading outages in practice, and it is normally composed of a transformer, a rectifier, a DC line and an inverter. Actually, the rectifier terminal can be regarded as a bus with real power consumption $P_{r}$, while the inverter terminal can be treated as a bus with real power generation $P_{i}$. The direct current from the rectifier to the inverter is computed as follows \cite{kun94}
$$
I_d=\frac{3\sqrt{3}(\cos\alpha-\cos\gamma)}{\pi(R_{cr}+R_L-R_{ci})},
$$
where $\alpha\in[\pi/30,\pi/2]$ denotes the ignition delay angle of the rectifier, and $\gamma\in[\pi/12,\pi/9]$ represents the extinction advance angle of the inverter. $R_{cr}$ and $R_{ci}$ refer to the equivalent communicating resistances for the rectifier and inverter, respectively. Additionally, $R_L$ denotes the resistance of the DC transmission line. Thus the power consumption at the rectifier terminal is
\begin{equation}\label{pr}
    P_r=\frac{3\sqrt{3}}{\pi}I_d\cos\alpha-R_{cr}I^2_d,
\end{equation}
and at the inverter terminal is
\begin{equation}\label{pi}
    P_i=\frac{3\sqrt{3}}{\pi}I_d\cos\gamma-R_{ci}I^2_d=P_r-R_LI^2_d.
\end{equation}
Note that $P_r$ and $P_i$ keep unchanged when $\alpha$ and $\gamma$ are fixed.

\subsection{Protective relay}
The protective relays are indispensable components in power systems protection and control. When the power flow exceeds the given threshold of the branch, the timer of circuit breaker starts to count down from the preset time \cite{jia16}. Once the timer runs out of the preset time, the transmission line is severed by circuit breakers and its branch admittance becomes zero. Specifically, a step function is designed to reflect the physical characteristics of branch outage as follows
$$
g(P_{e,i},\sigma_i)=\left\{
                            \begin{array}{ll}
                              0, & \hbox{$|P_{e,i}|>\sigma_i$ and $t_c>T$;} \\
                              1, & \hbox{otherwise.}
                            \end{array}
                          \right.
$$
where $T$ is the preset time of the timer in protective relays, and $t_c$ denotes the counting time of the timer. In addition, $P_{e,i}$ denotes the power flow on Branch $i$ with the threshold $\sigma_i$.

\begin{IEEEbiography}{Chao Zhai}
\end{IEEEbiography}

\begin{IEEEbiographynophoto}{Gaoxi Xiao}
\end{IEEEbiographynophoto}

\begin{IEEEbiographynophoto}{Hehong Zhang}
\end{IEEEbiographynophoto}

\begin{IEEEbiographynophoto}{Tso-Chien Pan}
\end{IEEEbiographynophoto}

\end{document}